\documentclass[11pt]{article}
\usepackage[utf8]{inputenc}
\usepackage{hyperref}
\usepackage{amsmath}
\usepackage{amsthm}
\usepackage{amssymb}
\usepackage{authblk}
\usepackage{cite}
\usepackage{lipsum}
\usepackage{geometry}
\usepackage{enumerate}
\usepackage{array}
\usepackage{multicol}
\usepackage{multirow}
\usepackage{xcolor}
\usepackage{kotex}
\usepackage{makecell}
\newtheorem{theorem}{Theorem}[section]
\numberwithin{theorem}{section}
\newtheorem{lemma}[theorem]{Lemma}
\newtheorem{defi}[theorem]{Definition}
\newtheorem{corollary}[theorem]{Corollary}

\newtheorem{remark}[theorem]{Remark}

\newcommand{\Max}{\displaystyle \max}
\newcommand{\Sum}{\displaystyle \sum}
\newcommand{\Frac}{\displaystyle \frac}
\newcommand{\F}{\mathbb F}

\newcommand{\Fp}{\mathbb{F}_{p}}
\newcommand{\Fq}{\mathbb{F}_{q}}

\newcommand{\Fqmul}{\mathbb{F}_{q}^*}

\begin{document}

\title{Locally-APN Binomials with Low Boomerang Uniformity in Odd Characteristic}
\author{ Namhun Koo$^1$, Soonhak Kwon$^{2,3}$, Minwoo Ko$^2$, Byunguk Kim$^2$\\
	\small{\texttt{ Email: komaton@skku.edu, shkwon@skku.edu, minwoo1403@skku.edu, kbu0923@g.skku.edu}}\\
	\small{$^1$Institute of Basic Science, Sungkyunkwan University, Suwon, Korea}\\
	\small{$^2$Department of Mathematics, Sungkyunkwan University, Suwon, Korea}\\
	\small{$^3$Applied Algebra and Optimization Research Center, Sungkyunkwan University, Suwon, Korea}
}
\date{}

\maketitle

\begin{abstract}
	Recently, several studies have shown that when $q\equiv3\pmod{4}$, 
	for certain choices of $r$, the function $F_r(x)=x^r+x^{r+\frac{q-1}{2}}$ defined over $\Fq$ is locally-APN and has boomerang uniformity at most~$2$. In this paper, we extend these results by showing that if there is at most one $x\in \Fq$ with $\chi(x)=\chi(x+1)=1$ satisfying $(x+1)^r - x^r = b$ for all $b\in \Fqmul$ and $\gcd(r,q-1)\mid 2$, then $F_r$ is locally-APN with boomerang uniformity at most $2$. Moreover, we study the differential spectra of $F_3$ and $F_{\frac{2q-1}{3}}$, and the boomerang spectrum of $F_2$ when $p=3$.
	
	
\end{abstract}

\section{Introduction}

In this paper, let $q=p^n$ be an odd prime power and $\Fq$ be the finite field of $q$ elements. 
In recent years, considerable attention has been devoted to the study of the differential properties of functions of the form
\begin{equation}
	F_{r,u}(x)=x^r(1+u\chi(x)),
\end{equation}
where $\chi(x)=x^{\frac{q-1}{2}}$ is the quadratic character over $\Fq$, and $u \in \Fqmul$. Differential uniformity, introduced by Nyberg\cite{Nyb94}, quantifies the resistance of an S-box in a block cipher to differential cryptanalysis. The definition of the differential uniformity is given below.

\begin{defi}
	Let $F$ be a function on $\Fq$. Define $\delta_F(a,b)$ by the number of solutions of $F(x+a)-F(x)=b$, where $a\in \Fqmul$ and $b\in \Fq$. Then, the \textbf{differential uniformity} of $F$ is defined by :
	\begin{equation*}
		\delta_F = \Max_{a\in \Fqmul, b\in \Fq}\delta_F(a,b).
	\end{equation*} 
\end{defi}
A function on $\Fq$ is said to be \emph{perfect nonlinear (PN)}, if its differential uniformity is equal to $1$. When its differential uniformity equals~$2$, the function is called \emph{almost perfect nonlinear (APN)}.
In~\cite{NH07}, Ness and Helleseth showed that $F_{r,u}$ is APN when $r=q-2$ and $\chi(u-1)=\chi(u+1)=\chi(u)$ for $p=3$. 
Subsequently, Zeng \emph{et al.} \cite{ZHYJ07} extended this result to the case $q\equiv 3\pmod{4}$, 
proving that $F_{q-2,u}$ is APN if $\chi(u-1)=\chi(u+1)=-\chi(5u\pm3)$. 

More recently, several studies have focused on the non-APN case of $F_{q-2,u}$. 
For instance, Xia \emph{et al.}~\cite{XBC+24} investigated the differential uniformity of $F_{q-2,u}$ 
and determined its differential spectrum when $\chi(u+1)=\chi(u-1)$ in characteristic~$3$. 
Their results were later extended in~\cite{XLB+25} to the general case where $q\equiv 3 \pmod{4}$. 
Ren \emph{et al.}~\cite{RXY26} studied the differential spectrum of $F_{q-2,u}$ in the complementary case $\chi(u+1)\ne \chi(u-1)$, 
while Lyu \emph{et al.}~\cite{LWZ24} independently showed that $F_{q-2,u}$ is a differentially $4$-uniform permutation under the same condition.

Budaghyan and Pal~\cite{BP25} conjectured, based on experimental evidence, 
that there might exist an infinite subclass of APN functions among $F_{2,u}$. 
However, this conjecture was disproved in~\cite{MW25}, 
where the authors conducted a detailed study of the differential uniformity of $F_{2,u}$. 
The same conjecture was independently disproved in~\cite{BS25} 
by means of function field theory. 
Koo and Kwon~\cite{KK26} further investigated the differential uniformity of $F_{r,u}$ 
for $r=\frac{q+1}{4}$. The differential properties of $F_{r,u}$ are summarized in Table~\ref{table_DU_Fru}. Here, the fifth column (“Spectrum”) indicates whether the corresponding spectrum is explicitly determined (“O”) or not (“X”).

\renewcommand{\arraystretch}{1.3}
\begin{table}[h!]
	\centering
	\begin{tabular}{cccccc}
		\hline
		$r$ & $u$ & Conditions & $\delta_{F_{r,u}}$
		& Spectrum & Ref. \\
		\hline
		$q-2$ & $\pm1$ &$p=3$, $n>3$ : odd & $\frac{q+1}{4}$ (locally-PN) & O & \cite{LWZ24} \\
		$q-2$ & $\pm1$ &$q\equiv 3\pmod{4}$,\ $p\ne3$ & $\frac{q+1}{4}$ (locally-APN) & O & \cite{LWZ24} \\
		$q-2$ & $u\neq\pm1$  &\makecell{
			$q\equiv 3 \pmod{4}$, $q\ne 3$ \\
			$\chi(1+u)=\chi(u-1)=-\chi(5u\pm 3)$
		} & 2 (APN) & O & \cite{ZHYJ07,XLB+25} \\
		$q-2$ & $u\neq\pm1$ &\makecell{
			$q\equiv 3 \pmod{4}$ \\
			$\chi(1+u)=\chi(u-1) \ne-\chi(5u\pm 3)$
		} & $ 3 $ & O & \cite{XLB+25} \\
		$q-2$ & $u\neq\pm1$ &\makecell{
			$q\equiv 3 \pmod{4}$  \\
			$\chi(1+u)\ne\chi(u-1) $
		} & $\delta_{F_{r,u}} \leq4$ & O & \cite{XBC+24,RXY26} \\
		$2$ & $\pm1$ & $q \equiv3 \pmod{4}$,  & $\frac{q+1}{4}$ (locally-APN) & X & \cite{MW25}\\
		$2$ & $u\neq\pm1$ & $q \equiv3 \pmod{4}$ & $3\leq\delta_{F_{r,u}} \leq5$ & X & \cite{MW25}\\
		$\frac{q+1}{4}$ & $\pm1$ & $q \equiv3 \pmod{8}$ & $\frac{q+1}{4}$ (locally-PN) & O & \cite{KK26}\\
		$\frac{q+1}{4}$ & $\pm1$ & $q \equiv7 \pmod{8}$ & $\frac{q+1}{4}$ (locally-APN) & O & \cite{KK26}\\
		$\frac{q+1}{4}$ & $u\neq\pm1$ & $q \equiv3 \pmod{4}$ & $\delta_{F_{r,u}} \leq5$ & X & \cite{KK26}\\
		$p^k+1$ & $\pm1$ & $n$ : odd & $\frac{q+1}{4}$ (locally-APN) & X & This paper \\
		$\frac{3^k+1}{2}$ & $\pm1$ & $p=3$, $n$ : odd, $\gcd(k,n)=1$ & $\frac{q+1}{4}$ (locally-APN) & X & This paper \\
		$3$ & $\pm1$ & $q \equiv 11 \pmod{12}$ & $\frac{q+1}{4}$ (locally-APN) & O & This paper \\
		$\frac{2q-1}{3}$ & $\pm1$ & $q \equiv 11 \pmod{12}$ & $\frac{q+1}{4}$ (locally-APN) & O & This paper \\
		$\frac{3^{\frac{n+1}{2}} - 1}{2}$& $\pm1$ & $p=3,\; n$ : odd  & $\frac{q+1}{4}$ (locally-APN) & X & This paper \\
		$\frac{3^{n+1}-1}{8}$ & $\pm1$ & $p=3,\; n$ : odd & $\frac{q+1}{4}$ (locally-APN) & X & This paper \\
		\hline
	\end{tabular}
	\caption{Differential uniformity of $F_{r,u}$}\label{table_DU_Fru}
\end{table}

The boomerang uniformity was introduced as the maximum value of the \emph{boomerang connectivity table} (BCT), 
a tool proposed to analyze the \emph{boomerang attack}, 
which is a variant of differential cryptanalysis~\cite{CHP+18}. 
Originally, this notion was defined only for permutations over finite fields of characteristic~$2$. 
Later, Li \emph{et al.}~\cite{LQSL19} extended it to general functions, not necessarily permutations, 
so that it can be applied to mappings in arbitrary characteristic.

\begin{defi}
	Let $F$ be a function on $\Fq$. Define $\beta_F(a,b)$ by the number of common solutions $(x,y)$ of the following system :
	\begin{equation*}
		\begin{cases}
			F(x)-F(y)=b,\\
			F(x+a)-F(y+a)=b.
		\end{cases}
	\end{equation*}
	Then the \textbf{boomerang uniformity} of $F$ is given by :
	\begin{equation*}
		\beta_F = \Max_{a,b\in \Fqmul}\beta_F(a,b).
	\end{equation*} 
\end{defi}

For the case $u\ne\pm1$, the existing studies have focused only on the differential properties of $F_{r,u}$, 
and no results on their boomerang behavior have been reported so far. 
In contrast, when $u=\pm1$, both the differential and boomerang properties of $F_{r,u}$ have been investigated in several recent works.
In the following, we concentrate on this particular case. 
Note that $F_{r,-1}$ is affine equivalent to $F_{r,1}$. 
In this case, we simply write
\[
F_r(x)=x^r(1+\chi(x))=x^r+x^{r+\frac{q-1}{2}}.
\]
Recently, several studies have investigated the differential and boomerang spectra of $F_r$
for various choices of $r$. These works consistently show that $F_r$ is locally-APN 
and has boomerang uniformity at most~$2$ in all known cases. 
Lyu \emph{et al.}~\cite{LWZ24} determined the differential and boomerang spectra of $F_{q-2}$, 
showing that it is locally-PN with boomerang uniformity~$0$ when $p=3$, 
and locally-APN with boomerang uniformity at most~$1$ when $p\ne3$. 
Mesnager and Wu~\cite{MW25} proved that $F_2$ is locally-APN with boomerang uniformity at most~$2$, 
and further investigated its differential spectrum. 
They also proved that the boomerang uniformity of $F_2$ equals~$2$, when $q\ge9613^2$. 
However, we will show in Section~\ref{subsec_BS_F2} that the boomerang uniformity of $F_2$ equals~$1$ when $p=3$. 
Ren and Yan~\cite{RY26} independently analyzed the differential spectrum of $F_2$ using a different approach. 
Koo and Kwon~\cite{KK26} studied $F_r$ for $r=\frac{q+1}{4}$, 
showing that $F_r$ is locally-PN with boomerang uniformity~0 when $q\equiv3\pmod{8}$, 
and locally-APN with boomerang uniformity~$2$ when $q\equiv7\pmod{8}$. 
They also provided complete characterizations of the differential and boomerang spectra of $F_r$ in these cases. The boomerang properties of $F_{r}$ are summarized in Table~\ref{table_BU_Fr}. Here, “O” indicates that the corresponding spectrum is explicitly determined, whereas “X” indicates that it is not.

\begin{table}[h!]
	\centering
	\begin{tabular}{ccccc}
		\hline
		$r$ &  Conditions & $\beta_{F_r}$ & Spectrum & Ref. \\
		\hline
		$q-2$ & $p=3$, $n$ : odd & $0$ & - & \cite{LWZ24} \\
		$q-2$ & $q \equiv3 \pmod{4}$, $q\ge 7$ & $1$  & O & \cite{LWZ24} \\
		$2$ & $q \equiv3 \pmod{4}$, large $q$ & 2 & X & \cite{MW25} \\
		$\frac{q+1}{4}$ & $q\equiv3\pmod8$ & 0  & - & \cite{KK26} \\
		$\frac{q+1}{4}$ & $q\equiv7\pmod8$ & 2  & O & \cite{KK26}\\
		$p^k+1$ & $n$ : odd & $\le2$ & X & This paper \\
		$\frac{3^k + 1}{2}$ & $p=3,\; n$ : odd, $\gcd(k,n)=1$ & $\le2$& X & This paper \\
		$3$ & $q \equiv 11 \; (\bmod 12)$ &$\le2$ & X & This paper \\
		$\frac{2q-1}{3}$ & $q \equiv 11 \; (\bmod 12)$ &$\le2$ & X & This paper \\
		$\frac{3^{n+1}-1}{2}$ & $p=3,\; n$ : odd &$\le2$& X & This paper \\
		$\frac{3^{n+1}-1}{8}$ & $p=3,\; n$ : odd &$\le2$ & X & This paper \\
		$2$ & $p=3, \; n\geq3$ : odd & 1 & O & This paper\\
		\hline
	\end{tabular}
	\caption{Boomerang uniformity of $F_{r}$.}\label{table_BU_Fr}
\end{table}

In this paper, we study the differential and boomerang properties of 
$F_r$ when $q\equiv 3 \pmod{4}$. 
In particular, we show that if the equation $(x+1)^r - x^r = b$ has at most one solution satisfying 
$\chi(x)=\chi(x+1)=1$ and $\gcd(r,q-1)\mid 2$, 
then $\delta_{F_r}(1,b)\le 2$ for all $b\in\Fqmul$. 
This implies that $F_r$ is locally-APN. 
Furthermore, we also prove that its boomerang uniformity is at most~$2$.
The above condition is more likely to hold when the differential uniformity of $x^r$ is small, and it certainly holds when $x^r$ is PN. 
In this paper, we show that this is indeed the case for the exponents $r$ listed in Table~\ref{table_r} in Section \ref{sec_pre}. 
In addition, we further investigate the differential spectrum for $r=3$ and $r=\frac{2q-1}{3}$, 
and the boomerang spectrum of $F_2$ when $p=3$. 
By Lemma~\ref{CM linear equiv}, the latter result can be extended to the case $k=n-1$ 
for $r=\frac{3^k+1}{2}$.
For convenience, we recall that when $q\equiv 3 \pmod{4}$, one has $\chi(-1)=-1$, and that $x^{\frac{q+1}{4}}$ provides a square root of $x\in\Fq$ whenever $\chi(x)=1$. These facts will be used throughout the paper without further mention.

The rest of this paper is organized as follows. 
Section~\ref{sec_pre} provides the necessary preliminaries for our results. 
In Section~\ref{sec_Diff}, we investigate the differential properties of $F_r$. 
More precisely, for our exponents $r$ in Table \ref{table_r} we prove that $F_r$ is locally-APN in Section~\ref{subsec_DU_Fr}, 
and study the differential spectra of $F_3$ and $F_{\frac{2q-1}{3}}$ 
in Sections~\ref{subsec_DS_F3} and~\ref{subsec_DS_F3inv}, respectively. 
Section~\ref{sec_boomerang} focuses on the boomerang properties of $F_r$: 
we show that for our exponents $r$ in Table \ref{table_r} $F_r$ has boomerang uniformity at most~$2$ in Section~\ref{subsec_BU_Fr}, 
and determine the boomerang spectrum of $F_2$ for $p=3$ in Section~\ref{subsec_BS_F2}. 
Finally, Section~\ref{sec_con} concludes the paper.

\section{Preliminaries}\label{sec_pre}

When $F(x)=x^r$ is a power function, multiplying $\frac{1}{a^r}$ on both sides of $b=F(x+a)-F(x)=(x+a)^r-x^r$, we have
\begin{equation*}
	\Frac{b}{a^r}=\left(\Frac x a +1\right)^r -\left(\Frac x a\right)^r = (y+1)^r-y^r,
\end{equation*}
where $y=\Frac x a$. Thus, we have 
\begin{equation}\label{dupower_eqn}
	\delta_F(a,b)=\delta_F\left(1, \Frac{b}{a^r} \right), \text{ so }	\delta_F=\Max_{b\in \Fq}\delta_F(1,b),
\end{equation}
in this case. The notion of \emph{locally-APN} functions was originally introduced in \cite{BCC11} for power functions in the even characteristic case. In \cite{HLX+23}, Hu \emph{et al.} generalized it for odd characteristic cases, as follows. 
\begin{equation}\label{locallyapn_def}
	\delta_F(1,b)\le 2\text{ for all }b\in \Fq\setminus\Fp.
\end{equation}
In our previous paper \cite{KK26}, we extend the above definition to functions with the property
\begin{equation}\label{locallyapn_cordef_eqn}
	\delta_F(a,b)= \delta_F(1,g_a(b))\text{ for all }b\in \Fq,
\end{equation}
where $g_a$ permutes $\Fq$ for each $a\in\Fqmul$. 

\begin{defi}\label{locallyapn_cordef}
	Let $F$ be a function on $\Fq$ satisfying \eqref{locallyapn_cordef_eqn} for every $a\in \Fqmul$. Then, $F$ is called \textbf{locally-APN} if $\delta_F(1,b)\le 2$ for all $b\in \Fq\setminus \Fp$.
\end{defi}

\begin{lemma}\label{Fruproperty2_lemma}(Lemma 11 of \cite{MW25})
	Let $a\in \Fqmul$ and $b\in \Fq$. Then,
	\begin{align*}
		\delta_{F_{r,u}}(a,b)&=
		\begin{cases}
			\delta_{F_{r,u}}\left( 1, \frac{b}{a^r}\right) &\text{ if }\chi(a)=1,\\
			\delta_{F_{r,u}}\left( 1, \frac{b}{(-1)^{r+1}a^r}\right) &\text{ if }\chi(a)=-1,
		\end{cases}
		\\
		\beta_{F_{r,u}}(a,b)&=
		\begin{cases}
			\beta_{F_{r,u}}\left( 1, \frac{b}{a^r}\right) &\text{ if }\chi(a)=1,\\
			\beta_{F_{r,u}}\left( 1, \frac{b}{(-1)^{r}a^r}\right) &\text{ if }\chi(a)=-1.
		\end{cases}
	\end{align*}
\end{lemma}

We denote 
\begin{equation*}
	S_{ij}=\{x\in \Fq : \chi(x)=(-1)^i, \chi(x+1)=(-1)^j\},
\end{equation*}
where $i,j \in \{0,1\}$. The following lemma is well-known and will be used in Section \ref{subsec_DS_F3}.

\begin{lemma}\label{Sset num of elts}~\cite{Dic35} If $q\equiv 1\pmod{4}$, then $\#S_{00}=\frac{q-5}{4}$ and $\#S_{01}=\#S_{10}=\#S_{11}=\frac{q-1}{4}$. If $q\equiv 3\pmod{4}$, then $\#S_{00}=\#S_{10}=\#S_{11}=\frac{q-3}{4}$ and $\#S_{01}=\frac{q+1}{4}$.
\end{lemma}

The following four lemmas on the quadratic character are useful for deriving the differential spectrum of $F_3$ in Section \ref{subsec_DS_F3} and the boomerang spectrum of $F_2$ (for $p=3$) in Section \ref{subsec_BS_F2}.

\begin{lemma}\label{chi_lemma quad}(Theorem 5.48 of \cite{LN97})
	Let $f(x) = a_2x^2 +a_1 x + a_0 \in \Fq [x]$ with $p$ odd and $a_2 \ne 0$. Put $d=a_1^2-4a_0a_2$. Then,
	\begin{equation*}
		\sum_{x\in \Fq} \chi(f(x)) = 
		\begin{cases}
			-\chi(a_2) &\text{ if }d\ne 0,\\
			(q-1)\chi(a_2) &\text{ if }d = 0.
		\end{cases}
	\end{equation*}
\end{lemma}

\begin{lemma}\label{chi_lemma cubic}
	Let $q\equiv 3\pmod{4}$ and $f$ be a function on $\Fq$ with $f(-x)=-f(x)$ for all $x\in \Fq$. Then, $\Sum_{x\in \Fq}\chi(f(x))=0$.
\end{lemma}
\begin{proof}
	Applying $y=-x$, we have $$\Sum_{x\in \Fq}\chi(f(x)) = \Sum_{y\in \Fq}\chi(f(-y)) = \Sum_{y\in \Fq}\chi(-f(y)) = -\Sum_{y\in \Fq}\chi(f(y)),$$
	which completes the proof.
\end{proof}
By the above lemma, we get $\Sum_{x\in \Fq}\chi\left( x(x^2-1)\right) =0$, which will be used in Section \ref{subsec_BS_F2}.

\begin{lemma}\label{chi_lemma inequal}(Theorem 5.41 of \cite{LN97})
	Let $\eta(\cdot)$ be a multiplicative character of $\Fq$ of order $m>1$ and let $f\in\Fq[x]$ be a monic polynomial of positive degree that is not $m$-th power of a polynomial. Let $d$ be the number of distinct roots of $f$ in its splitting field over $\Fq$. Then for every $a\in\Fq$ we have
	\begin{equation*}
		\left| \sum_{x\in \Fq} \eta (af(x))\right| \le (d-1)\sqrt{q}.
	\end{equation*}
\end{lemma}

The following lemma is well-known and useful for our result.

\begin{lemma}\label{gcd_lemma}(Lemma 9 of \cite{EFR+20})
	Let $p$, $k$, $n$ be positive integers with $k\le n$. If $\frac{n}{\gcd(k,n)}$ is odd, then $\gcd(p^k+1, p^n-1)=2$.
\end{lemma}

In this paper, we show that for $r$ in Table \ref{table_r} $F_r$ is locally-APN with differential uniformity $\frac{q+1}{4}$ with boomerang uniformity at most $2$. These exponents $r$ form a subclass of those satisfying $\gcd(r,q-1)=2$, for which the monomial $x^r$ has low differential uniformity. This is relevant to our approach, since it tends to restrict the number of solutions of $(x+1)^r - x^r = b$ for each $b\in \Fqmul$.
 Note that the differential properties summarized in Table~\ref{table_r} 
are stated for the monomial $x^r$, 
but in some cases the equivalent form $x^{r+\frac{q-1}{2}}=x^r\chi(x)$ 
is actually PN or APN. 
Since $F_{r+\frac{q-1}{2}}(x)=x^{r+\frac{q-1}{2}}(1+\chi(x))=x^r(\chi(x)+1)=F_r(x)$ 
for all $x\in\Fq$, the two forms are identical as functions.

\renewcommand{\arraystretch}{1.3}
\begin{table}[h!]
	\begin{center}
		\begin{tabular}{cccc}
			\hline $r$ & Conditions & DU of $x^r$ & Ref. \\
			\hline $p^k+1$ & $n$ : odd & PN & \cite{CM97} \\
			$\frac{3^k+1}{2}$ & $p=3$, $n$ : odd, $\gcd(k,n)=1$ & PN & \cite{CM97} \\
			$3$ & $q\equiv 11\pmod{12}$ & APN & - \\
			$\frac{2q-1}{3}$ & $q\equiv 11\pmod{12}$ & APN & Inverse of $x^3$\\
			$\frac{3^{\frac{n+1}{2}}-1}{2}$ & $p=3$, $n$ : odd & APN & \cite{DMM+03} \\ 
			$\frac{3^{n+1}-1}{8}$ & $p=3$, $n$ : odd & APN & \cite{DMM+03} \\ 
			\hline
		\end{tabular}
		\caption{List of $r$ and Differential Properties of $x^r$.}\label{table_r}
	\end{center}
\end{table}
\renewcommand{\arraystretch}{1}

Moreover, when $r=\frac{3^k+1}{2}$ (the Coulter--Matthews exponent), 
$x^r$ is PN only for odd~$k$ \cite{CM97}. 
However, according to the linear equivalence established in Lemma~\ref{CM linear equiv}, 
our results hold for all~$k$.

Two functions $F$ and $F'$ defined on $\Fq$ are said to be \emph{linearly(affine) equivalent} 
if there exist linear(affine) permutations $L_1$ and $L_2$ over $\Fq$ such that 
$F' = L_1 \circ F \circ L_2$. 
It is well known that the differential and boomerang spectra are invariant under affine equivalence.

\begin{lemma}\label{CM linear equiv}
	Let $p$ be an odd prime and $n$ be odd. Then $F_{\frac{p^k+1}{2}}$ is linearly equivalent to $F_{\frac{p^{n-k}+1}{2}}$.
\end{lemma}
\begin{proof}
	Let $L(x) = x^{p^{n-k}}$. Then, for nonzero $x\in \Fq$,
	\begin{align*}
		(F_{\frac{p^k+1}{2}}\circ L)(x) &= \left( x^{p^{n-k}}\right) ^{\frac{p^k+1}{2}} (1+\chi(x^{p^{n-k}})) = x^{\frac{p^n+p^{n-k}}{2}} (1+\chi(x)) = x^{\frac{p^{n-k}+1}{2}} \cdot x^{\frac{p^{n}-1}{2}}(1+\chi(x)) \\
		&= x^{\frac{p^{n-k}+1}{2}}(\chi(x)+1) = F_{\frac{p^{n-k}+1}{2}}(x).
	\end{align*}
	Clearly, we can easily check that $(F_{\frac{p^k+1}{2}}\circ L)(0)=F_{\frac{p^{n-k}+1}{2}}(0) =0$.
\end{proof}

\section{Differential Properties}\label{sec_Diff}

In this section, we study the differential properties of $F_r$. We aim to count the number of solutions of 
\begin{equation}\label{DU_eqn}
	b=F_r(x+1)-F_r(x) = (x+1)^r(1+\chi(x+1))-x^r(1+\chi(x)).
\end{equation}
Let $D_{ij} (b)$ be the number of solutions of \eqref{DU_eqn} in $S_{ij}$ where $i,j\in\{0,1\}$. Then we have
\begin{equation*}
	\delta_{F_r} (1,b) = 
	\begin{cases}
		\sum D_{ij}(b) &\text{ if }b\ne 0,2,\\
		\sum D_{ij}(b) +1 &\text{ if }b= 0\text{ or }2,
	\end{cases}
\end{equation*}
 where the additional term $+1$ accounts for the solutions $x=0$ and $x=-1$, which correspond to $b=2$ and $b=0$, respectively.

\subsection{Differential Uniformity of $F_r$}\label{subsec_DU_Fr}

In this subsection, we show that $F_r$ is locally-APN for the exponents $r$ listed in Table~\ref{table_r}. First, we state the main theorem.

\begin{theorem}\label{DU_theorem}
	Let $r>1$ be an integer and $q=p^n \equiv 3 \pmod{4}$. If $(x+1)^r-x^r =b$ has at most $1$ solution in $S_{00}$ for all $b\in \Fqmul$ and $\gcd(r,q-1) \mid 2$, then $\delta_{F_r}(1,b)\le 2$ for all $b\in \Fqmul$ and
	\begin{equation}\label{Locally-APN DU}
		\delta_{F_r}=\delta_{F_r}(1,0)=\frac{q+1}{4},
	\end{equation}
	and hence $F_r$ is locally-APN. In particular, for each parameter set $(p,n,r)$ satisfying the corresponding conditions in Table~\ref{table_r}, $(x+1)^r-x^r =b$ has at most $1$ solution in $S_{00}$, and therefore, $F_r$ is locally-APN with \eqref{Locally-APN DU}.
\end{theorem}

Note that if $x^r$ is PN, then the assumption of the above theorem clearly holds. 
In this subsection, we prove Theorem~\ref{DU_theorem} and verify its assumption for the exponents $r$ listed in Table~\ref{table_r}. 
For the cases $r=3$ and $r=\frac{2q-1}{3}$, the verification of the assumption will be carried out separately in Sections~\ref{subsec_DS_F3} and \ref{subsec_DS_F3inv}, respectively, where we also determine the corresponding differential spectra. 
To handle the remaining cases, we analyze each $D_{ij}(b)$.

\begin{lemma}\label{DU_lemma S11}
	If $q\equiv 3 \pmod{4}$, then 
	\begin{equation*}
		D_{11} (b)=
		\begin{cases}
			\frac{q-3}{4}, &\text{ if }b=0,\\
			0, &\text{ otherwise.}
		\end{cases}
	\end{equation*}
\end{lemma}

\begin{proof}
	If $x\in S_{11}$, then \eqref{DU_eqn} implies $b=0$.
	By Lemma \ref{Sset num of elts}, we obtain $\#S_{11}=\frac{q-3}{4}$, as desired.
\end{proof}

\begin{lemma}\label{DU_lemma 00}
	Let $q\equiv 3 \pmod{4}$ and $b\in \Fq$. If $(x+1)^r-x^r =b$ has at most $1$ solution in $S_{00}$ for all $b\in \Fqmul$, then $D_{00}(b)\le 1$. If $\gcd(r,q-1)\in \{1,2\}$, then $D_{00}(0)=0$.
\end{lemma}
\begin{proof}
	If there is a solution $x\in S_{00}$, then we have 
	\begin{equation}\label{DU_eqn 00}
		(x+1)^r - x^r = \frac{b}{2}
	\end{equation}
	By assumption, we can easily see that $D_{00}(b)\le 1$.

	If there is a solution $x\in S_{00}$ in \eqref{DU_eqn} with $b=0$, then we have $(x+1)^r=x^r$. If $\gcd(r,q-1)=1$, then we have $x=x+1$, a contradiction. If $\gcd(r,q-1)=2$, then we obtain $x^2 = (x+1)^2$, which implies $x=-\frac{1}{2}$. Then, we get $1=\chi(x)\chi(x+1) = \chi\left( -\frac 1 4 \right) = \chi(-1)$, which is a contradiction since $\chi(-1)=-1$ when $q\equiv 3 \pmod{4}$. Therefore, $D_{00}(0)=0$.
\end{proof}

If $x^r$ is PN, then there is at most one such $x$ satisfying \eqref{DU_eqn 00}, we have $D_{00}(b)\le 1$.

\begin{lemma}\label{DU_lemma 01 zero}
	Let $q\equiv 3\pmod{4}$ and $b\in \Fq$. Then, $D_{01}(b) = 0$ or  $D_{10}(b) = 0$.
\end{lemma}
\begin{proof}
	If there is a solution $x\in S_{01}$ of \eqref{DU_eqn}, then $x^r = -\frac{b}{2}$. Since $x\in S_{01}$, we have $\chi\left( \frac{b}{2}\right)=\chi(-x^r)=-1$. By the way, if there is a solution $x'\in S_{10}$ of \eqref{DU_eqn}, then $(x'+1)^r = \frac{b}{2}$. Since $x'\in S_{10}$, we obtain $\chi\left( \frac{b}{2}\right)=\chi\left((x'+1)^r\right)=1$. Since $\chi\!\left(\frac{b}{2}\right) = 1$ and $\chi\!\left(\frac{b}{2}\right) = -1$ cannot occur simultaneously, we conclude that $D_{01}(b) = 0$ or $D_{10}(b) = 0$.
\end{proof}

\begin{lemma}\label{DU_lemma 01 le1}
	Let $q\equiv 3\pmod{4}$ and $b\in \Fq$. If $\gcd(r,q-1)\in \{1,2\}$, then $D_{01}(b) \le 1$ and  $D_{10}(b) \le 1$. Furthermore, if $b\in \{0,2\}$, $D_{01}(b)=D_{10}(b)=0$.
\end{lemma}
\begin{proof}
	 If there is a solution $x\in S_{01}$ of \eqref{DU_eqn}, then \eqref{DU_eqn} reduces to $x^r = -\frac{b}{2}$. If $\gcd(r,q-1)=1$, then there is at most one $x\in \Fq$ satisfying $x^r = -\frac{b}{2}$. If $\gcd(r,q-1)=2$, then there are at most $2$ solutions of $x^r = -\frac{b}{2}$. In this case, $r$ is even, and hence if $x=\theta$ is a solution of $x^r = -\frac{b}{2}$, then the other solution is $x=-\theta$. Since $\chi(\theta)\ne \chi(-\theta)$, at most one of $\theta$ and $-\theta$ belongs to $S_{01}$.
	
	When $b=0$, $x^r = -\frac{b}{2} =0$ implies $x=0$, and hence there is no solution in $S_{01}$. When $b=2$, if there is $x\in S_{01}$ satisfying $x^r = -\frac{b}{2} =-1$, then $\chi(x^r) = \left( \chi(x)\right) ^r = 1=\chi(-1)$, a contradiction. Therefore, $D_{01}(0)=D_{01}(2)=0$.
	
	By symmetry, the other proofs are similar to above, and we omit here. 
\end{proof}

We have now analyzed all the cases of $D_{ij}(b)$, 
and are ready to prove the main theorem of this subsection for all $r$ except $r=3$ and $r=\frac{2q-1}{3}$. 
It is worth noting that the lemmas established here also hold for these two exponents, 
and hence we have~\eqref{delta1b_inequal}, which will be proved in Lemmas~\ref{DU_lemma r=3 S00} and~\ref{DU_lemma r=3^(-1) S00}. 
These lemmas will also be used in the following subsections, 
where we investigate the differential spectra of $F_r$ for $r=3$ and $r=\frac{2q-1}{3}$.

\begin{proof}[Proof of Theorem \ref{DU_theorem}]
	If we substitute $x=0$ and $x=-1$ to \eqref{DU_eqn}, then we have $b=2$ and $b=0$, respectively. Thus, by lemmas in this subsection, we can see that $\delta_{F_r}(1,0) = \Frac{q+1}{4}$, and
	\begin{equation}\label{delta1b_inequal}
		\delta_{F_r}(1,b)\le D_{00}(b) +1,\ \ \text{ if }b\in \Fqmul.  
	\end{equation}
	Thus, the first part of Theorem~\ref{DU_theorem} is proved, namely that $F_r$ is locally-APN with $\delta_{F_r}=\delta_{F_r}(1,0)=\frac{q+1}{4}$ under the given assumptions. 
	We now turn to the second part, where we verify that the condition $(x+1)^r - x^r = b$ has at most one solution in $S_{00}$ for the exponents $r$ listed in Table~\ref{table_r}. 
	It is clear that $(x+1)^r - x^r = b$ has at most $1$ solution in $S_{00}$ for all $b\in \Fqmul$, when $x^r$ is PN. 
	For $r\in\left\{\frac{3^{\frac{n+1}{2}}-1}{2},\ \frac{3^{n+1}-1}{8}\right\}$, 
	it was shown in~\cite{DMM+03} that the equation $(x+1)^r - x^r = b$ admits at most one solution in each of the sets $S_{00}\cup S_{11}$ and $S_{01}\cup S_{10}$. 
	Consequently, the condition $D_{00}(b)\le 1$ follows directly from their argument.
\end{proof}

\subsection{Differential Spectrum when $r=3$}\label{subsec_DS_F3}

In this subsection, we study the differential spectrum of $F_3$ when $q\equiv 11\pmod{12}$. We wish to compute the number of solutions of 
\begin{equation}\label{DU_eqn r=3}
	b=F_3(x+1)-F_3(x) = (x+1)^3(1+\chi(x+1))-x^3(1+\chi(x)).
\end{equation}

We already observe that $\delta_{F_3}(1,0)=\frac{q+1}{4}$, in Section \ref{subsec_DU_Fr}. By Lemmas \ref{DU_lemma S11}, \ref{DU_lemma 01 zero} and \ref{DU_lemma 01 le1}, we have
\begin{equation}\label{delta1b_eqn}
	\delta_{F_r}(1,b) = 
	\begin{cases}
		D_{00}(b) + 1 & \text{ if }b=2,\\
		D_{00}(b) + D_{01}(b) + D_{10}(b)  & \text{ otherwise, }
	\end{cases}
\end{equation}
for $b\in \Fqmul$. We first characterize $D_{00}(b)$, $D_{01}(b)$ and $D_{10}(b)$.

\begin{lemma}\label{DU_lemma r=3 S00}
	Let $r=3$ and $q\equiv 3\pmod{4}$. Then,
	\begin{equation*}
		D_{00} (b)=
		\begin{cases}
			1, &\text{ if }\chi(2b-1)=\chi(2(b-2))=\chi(3),\\
			0, &\text{ otherwise.}
		\end{cases}
	\end{equation*}
\end{lemma}

\begin{proof}
	If $x\in S_{00}$, then \eqref{DU_eqn r=3} becomes $\frac{b}{2} = (x+1)^3-x^3 = 3x^2 + 3x +1$, so we get 
	\begin{equation}\label{DU_eqn1 r=3 S00}
		x^2+x - \frac{b-2}{6}=0.
	\end{equation}
	The discriminant of \eqref{DU_eqn1 r=3 S00} is $\frac{2b-1}{3}$.
	Hence, \eqref{DU_eqn1 r=3 S00} has two solutions if and only if $\chi(2b-1)=\chi(3)$. If \eqref{DU_eqn1 r=3 S00} has two solutions $x_1$ and $x_2$, then $x_1x_2=- \frac{b-2}{6}$. 
	Since $x \in S_{00}$, we have $\chi(x(x+1))=\chi\left( \frac{b-2}{6}\right)=1$. Since $\chi(x_1x_2) = \chi\left( -\frac{b-2}{6}\right)=-1$, \eqref{DU_eqn1 r=3 S00} has at most one solution in $S_{00}$. 
	Without loss of generality, we assume that $\chi(x_1)=1$ and $\chi(x_2)=-1$. Since $\chi(6(b-2))=\chi(x_1(x_1+1))=1$ from \eqref{DU_eqn1 r=3 S00}, we have $\chi(x_1+1)=1$ so $x_1 \in S_{00}$.
	
	If $b=\frac{1}{2}$, then \eqref{DU_eqn1 r=3 S00} has one solution $x=-\frac{1}{2}$. But, we have $\chi(x+1) = \chi\left( \frac{1}{2} \right) \ne \chi\left( -\frac{1}{2} \right)=\chi(x)$ in this case, and hence $-\frac{1}{2}\not \in S_{00}$. So, $D_{00} \left(\frac{1}{2} \right)=0$. 
\end{proof}

Next, we characterize all $b$ satisfying $D_{01}(b) + D_{10}(b) = 1$.
Note that we have already shown in Lemma \ref{DU_lemma 01 le1} that $D_{01} (b)\le 1$ and $D_{10} (b)\le 1$.

\begin{lemma}\label{DU_lemma r=3}
	Let $r=3$ and $q\equiv 11\pmod{12}$. Then, $D_{01} (b) + D_{10} (b) = 1$ if and only if $\chi\left( b\left( b^{\frac{2q-1}{3}} - \alpha\right) \right) =-1$, where $\alpha = 2^{\frac{2q-1}{3}}$. 
\end{lemma}

\begin{proof}
	If $x\in S_{01}$, then \eqref{DU_eqn r=3} becomes 
	\begin{equation}\label{DU_eqn1 r=3}
		x^3=-\frac{b}{2}.
	\end{equation}
	Since $x\in S_{01}$, 
	\begin{equation}\label{DU_eqn2 r=3}
		1=\chi(x)=\chi(x^3) = \chi\left( -\frac{b}{2}\right)
	\end{equation}
	implies that $\chi(2b)=-1$. Since $\frac{2q-1}{3}$ is odd, raising both sides of \eqref{DU_eqn1 r=3} to the $\frac{2q-1}{3}$-th power, we have
	\begin{equation*}
		x=\left( -\frac{b}{2}\right) ^{\frac{2q-1}{3}} = -\left( \frac{b}{2}\right) ^{\frac{2q-1}{3}}.
	\end{equation*}
	Consequently, we have $\chi(x) = \chi\left( -\frac {b}{2}\right) =-\chi(2b)=1$. Moreover,
	$$\chi(x+1)=\chi\left( -\left( \frac{b}{2}\right) ^{\frac{2q-1}{3}} + 1\right)= \chi\left( -\frac{b^{\frac{2q-1}{3}} - 2^{\frac{2q-1}{3}}}{2^{\frac{2q-1}{3}}}\right)=\chi\left( -\frac{b^{\frac{2q-1}{3}} - 2^{\frac{2q-1}{3}}}{2}\right)=-1$$
	if and only if $\chi\left( -2\left( b^{\frac{2q-1}{3}} -\alpha \right) \right)=-1$. Therefore, we have $D_{01} (b)=1$ if and only if 
	\begin{equation*}
		\chi(2b)=\chi\left( -2\left( b^{\frac{2q-1}{3}} - \alpha\right)  \right)=-1.
	\end{equation*}
	We can similarly show that $D_{10} (b)=1$ if and only if 
	\begin{equation*}
		\chi(2b)=1\text{ and }\chi\left( 2\left( b^{\frac{2q-1}{3}} - \alpha\right) \right) =-1.
	\end{equation*}
	We can observe that $D_{01} (b)=1$ and $D_{10} (b)=1$ cannot occur simultaneously, and $D_{01} (b) + D_{10} (b)=1$ if and only if
	\begin{equation*}
		-1 = \chi(2b) \chi\left( 2\left( b^{\frac{2q-1}{3}} - \alpha\right) \right) = \chi\left( b\left( b^{\frac{2q-1}{3}} - \alpha\right) \right),
	\end{equation*}
	which completes the proof.
\end{proof}

For any function $F$ satisfying \eqref{locallyapn_cordef_eqn}, the differential spectrum of $F$ is defined to be the multiset $DS_F = \{\omega_i : 0\le i \le \delta_F\}$, where
\begin{equation*}
	\omega_i = \#\{b\in \Fq : \delta_F(1,b)=i\}. 
\end{equation*} 
The following identity for the differential spectrum is well-known :
\begin{equation}\label{DS_identity}
	\sum_{i=0}^{\delta_F}\omega_i = \sum_{i=0}^{\delta_F}i\cdot \omega_i = q
\end{equation}
Next, we prove the main theorem of this subsection.
\begin{theorem}\label{DS_theorem r=3}
	Let $q\equiv 11 \pmod{12}$ with $q>11$. Then, $F_3(x) = x^3+x^{\frac{q+5}{2}}$ is locally-APN. Furthermore, the differential spectrum of $F_3$ is given by
	\begin{equation*}
		DS_{F_3} = \left\{ \omega_0=\frac{3}{8}(q-3)-\frac{\Gamma}{2} ,\ \omega_1 = \frac{1}{2}(q+1)+\Gamma,\ \omega_2 = \frac{1}{8}(q-3)-\frac{\Gamma}{2},\ \omega_{\frac{q+1}{4}}=1 \right\},
	\end{equation*}
	where $\Gamma = \Sum_{x\in S}\chi(x)\chi(x-\alpha)$ with $\alpha = 2^{\frac{2q-1}{3}}$ and $S=\left\{ x\in \Fq : \chi(2x^3-1)=\chi(2x^3-4)=1\right\}$.
\end{theorem}

\begin{proof}
	Combining \eqref{delta1b_inequal} with Lemma \ref{DU_lemma 00}, we can see that $\delta_{F_3}(1,b) \le 2$ for all $b\in \Fqmul$. 
	
	Next, we count the number of $b\in \Fqmul$ satisfying $\delta_{F_3}(1,b)=2$. By Lemma \ref{DU_lemma r=3 S00}, we have $D_{00}(2)=0$, and hence $\delta_{F_3}(1,2)=1$, using \eqref{delta1b_eqn}. Hence, $\delta_{F_3}(1,b)=2$ if and only if $D_{00} (b)=1$ and $D_{01} (b) + D_{10} (b)=1$. Since $q\equiv 11 \pmod{12}$, we have $\chi(3)=1$. Thus, by Lemmas \ref{DU_lemma r=3 S00} and \ref{DU_lemma r=3}, $\delta_{F_3}(1,b)=2$ if and only if $\chi(2b-1)=\chi(2b-4)=1$ and $\chi\left( b\left( b^{\frac{2q-1}{3}} - \alpha\right) \right) =-1$.
	Hence, we have
\begin{align*}
	8\omega_2=& \sum_{\substack{x\in \Fq \\ x\neq 0,\,\frac{1}{2},\,2}} (1+\chi(2x-1)) (1+\chi(2x-4)) \left(1-\chi\left( x\left( x^{\frac{2q-1}{3}} -\alpha\right)\right)\right) \\
	=& -2+\Sum_{x\in \Fq} (1+\chi(2x-1)) (1+\chi(2x-4)) \left(1-\chi\left( x\left( x^{\frac{2q-1}{3}} -\alpha\right)\right)\right) \\
	=&-2+D_1 -D_2 ,
\end{align*}
where
\begin{align*}
	D_1 &= \Sum_{x\in \Fq}(1+\chi(2x-1))(1+\chi(2x-4)),\\
	D_2 &= \Sum_{x\in \Fq}(1+\chi(2x-1))(1+\chi(2x-4))\chi\left( x\left( x^{\frac{2q-1}{3}} -\alpha\right)\right).
\end{align*}
By Lemma \ref{chi_lemma quad},
\begin{equation*}
	D_1 =\Sum_{x\in \Fq}(1+\chi(2x-1))(1+\chi(2x-4)) = q + \Sum_{x\in \Fq} \chi(2x-1) +\Sum_{x\in \Fq}\chi(2x-4) + \Sum_{x\in \Fq}\chi((2x-1)(2x-4)) = q-1.
\end{equation*}
Observe that
\begin{align*}
	D_2 &=\Sum_{x\in \Fq} \chi(x)\chi\left( x^{\frac{2q-1}{3}} -\alpha\right) (1+\chi(2x-1)) (1+\chi(2x-4))\\
	&\underset{x=u^3}{=}\Sum_{u\in \Fq} \chi(u^3)\chi\left( u -\alpha\right)(1+\chi(2u^3-1))(1+\chi(2u^3-4)) \\
	&=\Sum_{u\in \Fq} \chi(u)\chi\left( u -\alpha\right)(1+\chi(2u^3-1))(1+\chi(2u^3-4)) = 4\Sum_{u \in S} \chi(u)\chi\left( u -\alpha\right).
\end{align*}
where $S = \left\{ x\in \Fq : \chi(2x^3-1)=\chi(2x^3-4)=1\right\}$. Note that if $u^3=\frac{1}{2}$, then $1+\chi(2u^3-4)=1-\chi(3)=0$, and if $u^3=2$, then $u=\alpha$ and hence $\chi(u-\alpha)=0$, so these cases do not contribute to the sum. Thus, 
\begin{align*}
	\omega_2 =\frac{1}{8}\left( -2+D_1 -D_2 \right)  = \frac{1}{8}\left( q-3 -4\Sum_{x\in S}\chi(x)\chi(x-\alpha)\right).
\end{align*}
We get $\omega_0=\frac{3}{8}(q-3)-\frac{\Gamma}{2}$ and $\omega_1 = \frac{1}{2}(q+1)+\Gamma$ from \eqref{DS_identity} with the above and $\omega_{\frac{q+1}4}=1$.

To complete the proof, it remains to show that $\omega_2>0$. Since
\begin{align*}
	D_2 &= 4\Sum_{u \in S} \chi(u)\chi\left( u -\alpha\right)=\Sum_{u\in \Fq} \chi(u)\chi\left( u -\alpha\right)(1+\chi(2u^3-1))(1+\chi(2u^3-4)) \\
	&= -1+\Sigma_1+\Sigma_2+\Sigma_3,
\end{align*}
where
\begin{align*}
	\Sigma_1&=\sum_{u\in\F_q}\chi\bigl(u(u-\alpha)(2u^3-1)\bigr),\\
	\Sigma_2&=\sum_{u\in\F_q}\chi\bigl(u(u-\alpha)(2u^3-4)\bigr),\\
	\Sigma_3&=\sum_{u\in\F_q}\chi\bigl(u(u-\alpha)(2u^3-1)(2u^3-4)\bigr).
\end{align*}
Since the corresponding polynomials have degrees $5,5,$ and $8$, respectively, and are not squares, Lemma~\ref{chi_lemma inequal} gives
$$
|\Sigma_1|\le 4\sqrt q,\qquad |\Sigma_2|\le 4\sqrt q,\qquad |\Sigma_3|\le 7\sqrt q.
$$
Hence,
$$
D_2\le -1+15\sqrt q,
$$
and therefore
$$
\omega_2=\frac18(q-3-D_2)\ge \frac18(q-2-15\sqrt q).
$$ 
In particular, $\omega_2>0$ for $q>231$. 

For the remaining values $q<231$, we verify by direct computation (using SageMath) that $\omega_2>0$ (see Table~\ref{table_ds}). 
This completes the proof.
\end{proof}

Note that the above theorem is confirmed via SageMath when $11\le q<100000$. Table \ref{table_ds} describes the differential spectrum of $F_3$ and $\Gamma$ in Theorem \ref{DS_theorem r=3}, computed via SageMath, when $q < 231$.

\renewcommand{\arraystretch}{1.1}
\begin{table}
	\begin{center}
		\begin{tabular}{ccc}
			\hline $q$ & $\Gamma$ & $DS_{F_3}$ \\
			\hline $11$ & $2$ & $\{ \omega_0 = 2,\ \omega_1 = 8,\ \omega_{3} = 1 \}$ \\
			$23$ & $-3$ & $\{\omega_0 = 9,\ \omega_1 = 9,\ \omega_2 = 4,\ \omega_{6} = 1 \}$\\
			$47$ & $-5$ & $\{ \omega_0 = 19,\ \omega_1 = 19,\ \omega_2 = 8,\ \omega_{12} = 1 \}$ \\
			$59$ & $2$ & $\{\omega_0 = 20,\ \omega_1 = 32,\ \omega_2 = 6,\ \omega_{15} = 1 \}$\\
			$71$ & $5$ & $\{ \omega_0 = 23,\ \omega_1 = 41,\ \omega_2 = 6,\ \omega_{18} = 1 \}$ \\
			$83$ & $-4$ & $\{\omega_0 = 32,\ \omega_1 = 38,\ \omega_2 = 12,\ \omega_{21} = 1 \}$\\
			$107$ & $-2$ & $\{ \omega_0 = 40,\ \omega_1 = 52,\ \omega_2 = 14,\ \omega_{27} = 1 \}$ \\
			$131$ & $-4$ & $\{\omega_0 = 50,\ \omega_1 = 62,\ \omega_2 = 18,\ \omega_{33} = 1 \}$\\
			$167$ & $13$ & $\{ \omega_0 = 55,\ \omega_1 = 97,\ \omega_2 = 14,\ \omega_{42} = 1 \}$\\
			$179$ & $8$ & $\{\omega_0 = 62,\ \omega_1 = 98,\ \omega_2 = 18,\ \omega_{45} = 1 \}$\\
			$191$ & $3$ & $\{ \omega_0 = 69,\ \omega_1 = 99,\ \omega_2 = 22,\ \omega_{48} = 1 \}$ \\
			$227$ & $0$ & $\{\omega_0 = 84,\ \omega_1 = 114,\ \omega_2 = 28,\ \omega_{57} = 1 \}$\\
			\hline
		\end{tabular}
		\caption{Differential spectrum $DS_{F_3}$ when $q < 231$ with $q \equiv 11 \pmod{12}$.}\label{table_ds}
	\end{center}
\end{table}
\renewcommand{\arraystretch}{1}

\subsection{Differential Spectrum when $r=\frac{2q-1}{3}$}\label{subsec_DS_F3inv}

In this subsection, we study the differential spectrum of $F_{\frac{2q-1}{3}}$ when $q\equiv 11\pmod{12}$. We wish to compute the number of solutions of 
\begin{equation}\label{DU_eqn r=3^(-1)}
	b=F_{\frac{2q-1}{3}}(x+1)-F_{\frac{2q-1}{3}}(x) = (x+1)^{\frac{2q-1}{3}}(1+\chi(x+1))-x^{\frac{2q-1}{3}}(1+\chi(x)).
\end{equation}
We already observe that $\delta_{F_{\frac{2q-1}{3}}}(1,0)=\frac{q+1}{4}$, in Section \ref{subsec_DU_Fr}. As in Section~\ref{subsec_DS_F3}, we derive the differential spectrum using~\eqref{delta1b_eqn}. We first characterize $D_{00}(b)$, $D_{01}(b)$ and $D_{10}(b)$.

\begin{lemma}\label{DU_lemma r=3^(-1) S00}
	Let $r=\frac{2q-1}{3}$ and $q\equiv 11\pmod{12}$. Then,
	\begin{equation*}
		D_{00} (b)=
		\begin{cases}
			1 &\text{ if }\chi(b(32-b^3))=1,\ \chi(b(b^3-8))=-1,\\
			0 &\text{ otherwise.}
		\end{cases}
	\end{equation*}
\end{lemma}

\begin{proof}
		If there exists $x\in S_{00}$ satisfying \eqref{DU_eqn r=3^(-1)}, then \eqref{DU_eqn r=3^(-1)} leads to 
	\begin{equation*}
		(x+1)^{\frac{2q-1}{3}} = x^{\frac{2q-1}{3}} + \frac{b}{2}.
	\end{equation*}
	Raising both sides of the above equation to the third power yields
	\begin{equation*}
		0=12bX^2+6b^2X+b^3-8,
	\end{equation*}
	where $X = x^{\frac{2q-1}{3}}$. The above equation has two solutions if and only if 
	\begin{equation*}
		(6b^2)^2 - 4\times 12b(b^3-8) = -12b^4 +384b = 12b(32-b^3)
	\end{equation*}
	is a square. Then, we obtain
	\begin{equation*}
		X=\frac{-6b^2\pm (12b(32-b^3))^{\frac{q+1}{4}}}{24b} = \frac{-3b^2\pm \left( 3b(32-b^3)\right) ^{\frac{q+1}{4}}}{12b}
	\end{equation*}
	Then, one has
	\begin{align*}
		x&=X^3 = \frac{\left( -3b^2\pm \left( 3b(32-b^3)\right) ^{\frac{q+1}{4}}\right) ^3}{1728b^3} \\
		&=  \frac{ -27b^6\pm 27b^4\left( 3b(32-b^3)\right) ^{\frac{q+1}{4}} - 27b^3(32-b^3)\pm 3b(32-b^3)\left( 3b(32-b^3)\right) ^{\frac{q+1}{4}}}{1728b^3}\\
		&=\frac{ -864b^3\pm 24b(4+b^3)\left( 3b(32-b^3)\right) ^{\frac{q+1}{4}}}{1728b^3}
		=\frac{ -36b^2\pm (4+b^3)\left( 3b(32-b^3)\right) ^{\frac{q+1}{4}}}{72b^2}.
	\end{align*}
	Hence, we can see that if we denote $x=x_1$ is a solution of \eqref{DU_eqn r=3^(-1)}, then the other is $x=x_2=-x_1 -1$, since the two solutions satisfy $x_1+x_2=-1$. If both $x_1$ and $x_2$ are in $S_{00}$, then $\chi(x_1+1)=\chi(-x_2)=-1$ and $\chi(x_2+1)=\chi(-x_1)=-1$, a contradiction. Hence, \eqref{DU_eqn r=3^(-1)} has at most one solution in $S_{00}$. Also, one can check that
	\begin{equation*}
		x(x+1)= -\frac{(b^3-8)^3}{2^6 \cdot 3^3\cdot b^3}
	\end{equation*}
	and $\chi(3)=1$, and therefore we have $\chi(x(x+1))=1$ if and only if $\chi(b(b^3-8))=-1$.
\end{proof}

Next, we characterize all $b$ satisfying $D_{01}(b) + D_{10}(b) = 1$.
Note that we have already shown in Lemma \ref{DU_lemma 01 le1} that $D_{01} (b)\le 1$ and $D_{10} (b)\le 1$.

\begin{lemma}\label{DU_lemma r=3^(-1)}
	Let $r=\frac{2q-1}{3}$ and $q\equiv 11\pmod{12}$. Then, 
	\begin{equation*}
		D_{01}(b)+D_{10}(b) = 
		\begin{cases}
			1 &\text{ if }\chi(b(b^3-8))=-1,\\
			0 &\text{ otherwise.}
		\end{cases}
	\end{equation*}
\end{lemma}

\begin{proof}
	If $x\in S_{01}$, then \eqref{DU_eqn r=3^(-1)} reduces to 
	\begin{equation}\label{DU_equation1 r=3inv}
		x^{\frac{2q-1}{3}} = -\frac{b}{2}.
	\end{equation}
	From the above equation, we get 
	$$\chi\left( \frac{b}{2}\right) = \chi\left( -x^{\frac{2q-1}{3}}\right) = -\chi(x)^{\frac{2q-1}{3}} = -1.$$  
	Raising both sides of \eqref{DU_equation1 r=3inv} to the third power yields
	$$x = -\frac{b^3}{8}.$$
	Since $x\in S_{01}$, we have
	$$1=\chi(-(x+1)) = \chi\left(  \frac{b^3-8}{8}\right) = \chi(2(b^3-8)).$$	
	Hence, $D_{01}(b)=1$ if and only if $\chi(2b) =-1$ and  $\chi(2(b^3-8))=1$. Similarly, $D_{10}(b)=1$ if and only if $\chi(2b) =1$ and  $\chi(2(b^3-8))=-1$. Thus, $\chi(2b)\chi(2(b^3-8))=-1$, 
	or equivalently $\chi(b(b^3-8))=-1$, 
	which yields the desired result.
\end{proof}

Next, we prove the main theorem of this subsection.
\begin{theorem}\label{DS_theorem r=3^(-1)}
	Let $q\equiv 11 \pmod{12}$ with $q>11$. Then, $F_{\frac{2q-1}{3}}(x) = x^{\frac{2q-1}{3}}+x^{\frac{q+1}{6}}$ is locally-APN. Furthermore, the differential spectrum of $F_{\frac{2q-1}{3}}$ is given by
	\begin{equation*}
		DS_{F_{\frac{2q-1}{3}}} = \left\{ \omega_0=\frac{q-3}{2},\ \omega_1 = \frac{q+5}{4},\ \omega_2 = \frac{q-3}{4},\ \omega_{\frac{q+1}{4}}=1 \right\}.
	\end{equation*}
\end{theorem}

\begin{proof}
	Combining \eqref{delta1b_eqn} with Lemma \ref{DU_lemma r=3^(-1) S00} and Lemma \ref{DU_lemma r=3^(-1)}, we obtain that $\delta(1,b)=2$ if and only if $\chi(b(32-b^3))=1$ and $\chi(b(b^3-8))=-1$. Hence, we have	
	\begin{align*}
		4\omega_2&=\Sum_{\substack{x\in \Fqmul \\ x\neq 2,\,2\alpha^2}}\left( 1+\chi(x(32-x^3))\right)  \left( 1-\chi(x(x^3-8))\right)  \\
		&=\Sum_{x\in \Fqmul}\left( 1+\chi(x(32-x^3))\right)  \left( 1-\chi(x(x^3-8))\right)-2.
	\end{align*}
	Applying $u=x^3$ and since $\frac{2q-1}{3}$ is odd,
	\begin{align*}
		4\omega_2&= \Sum_{u\in \Fqmul}\left( 1+\chi(u^{\frac{2q-1}{3}}(32-u))\right)  \left( 1-\chi(u^{\frac{2q-1}{3}}(u-8))\right)-2\\
		&=\Sum_{u\in \Fq}\left( 1+\chi(u(32-u))\right)  \left( 1-\chi(u(u-8))\right)-3.
	\end{align*}
	By Lemma \ref{chi_lemma quad}, 
	\begin{align*}
		4\omega_2&=\Sum_{u\in \Fq} 1 + \Sum_{u\in \Fq} \chi(u(32-u))- \Sum_{u\in \Fq}\chi(u(u-8)) - \Sum_{u\in \Fq}\chi(u^2(u-8)(32-u)) -3 \\
		&=q-1  - \Sum_{u\in \Fqmul}\chi(u^2(u-8)(32-u)) = q-1  - \Sum_{u\in \Fqmul}\chi((u-8)(32-u))\\
		&=q-2 - \Sum_{u\in \Fq}\chi((u-8)(32-u))=q-3.
	\end{align*}
	Hence, we obtain $\omega_2=\frac{q-3}{4}$. Applying $\omega_{\frac{q+1}{4}}=1$ and $\omega_2=\frac{q-3}{4}$ on \eqref{DS_identity}, we have $\omega_1=\frac{q+5}{4}$ and $\omega_0 = \frac{q-3}{2}$. We complete the proof.
\end{proof}

We note that the differential spectrum in this case coincides with that of 
$F_{\frac{q+1}{4}}$ when $q\equiv7\pmod{8}$ in~\cite{KK26}, 
although the two cases arise under different arithmetic conditions on~$q$. 
This theorem has also been verified numerically using SageMath for all $q$ with $11\le q<100000$.

\section{Boomerang Properties}\label{sec_boomerang}

In this subsection, we study boomerang properties of $F_r$. We consider the problem of finding the number of common solutions $(x,y)$ of the following system.
\begin{equation}\label{BU_system}
	\begin{cases}
		x^r(1+\chi(x))-y^r(1+\chi(y))=b,\\
		(x+1)^r(1+\chi(x+1))-(y+1)^r(1+\chi(y+1))=b.
	\end{cases}
\end{equation}
Denote by $B_{ijkl}(b)$ the number of solutions of \eqref{BU_system} in $S_{ij}\times S_{kl}$, where $i,j,k,l\in\{0,1\}$.

\subsection{Boomerang Uniformity of $F_r$}\label{subsec_BU_Fr}

In this subsection, we prove the following.
\begin{theorem}\label{BU_theorem}
	If $D_{00}(b)\le 1$ for all $b\in \Fqmul$ and $\gcd(r,q-1) \mid 2$, then $\beta_{F_r}\le 2$. In particular, for each parameter set $(p,n,r)$ satisfying the corresponding conditions in Table~\ref{table_r}, $\beta_{F_r}\le 2$.
\end{theorem}

For the exponents $r$ listed in Table~\ref{table_r}, we have already shown in Section~\ref{sec_Diff} that $D_{00}(b)\le 1$ for all $b\in \Fqmul$. 
Therefore, it suffices to prove the first assertion of the above theorem.
To prove the above theorem, we investigate each $B_{ijkl}(b)$.

\begin{lemma}\label{BU_lemma S11}
	Let $q\equiv 3\pmod{4}$ and $b\in \Fqmul$. If $\gcd(r,q-1)\in \{1,2\}$, then there is no solution $(x,y)$ of \eqref{BU_system} satisfying $x \in  (S_{11}\cup\{-1\})$ or $y \in  (S_{11}\cup\{-1\})$.
\end{lemma}
\begin{proof}
	By symmetry, we only show that there is no solution $(x,y)$ satisfying $x \in  (S_{11}\cup\{-1\})$. Suppose that there is such a solution $(x,y)$ satisfying $x \in  (S_{11}\cup\{-1\})$. If $y\in S_{00}$, then \eqref{BU_system} reduces to $y^r = (y+1)^r=-\frac{b}{2}$. This implies that $y^{\gcd(r,q-1)}=(y+1)^{\gcd(r,q-1)}$. If $\gcd(r,q-1)=1$, then we have $y=y+1$, a contradiction. If $\gcd(r,q-1)=2$, then we obtain $y^2 = (y+1)^2$, which implies $y=-\frac{1}{2}$. Then, we get $\chi(y) = \chi\left( -\frac 1 2 \right) = \chi \left( \frac 1 2 \right) = \chi(y+1) =1$, a contradiction.
	
	If $y\not \in S_{00}$, then at least one of equations in \eqref{BU_system} leads to $b=0$, a contradiction. 
\end{proof}

\begin{lemma}\label{BU_lemma nosol}
	Let $q\equiv 3 \pmod{4}$ and $b\in \Fqmul$. Then, $B_{0101}(b)=B_{0110}(b)=B_{1010}(b)=B_{1001}(b)=0$. 
\end{lemma}
\begin{proof}
	The proof is very similar with Lemma 3.8 of our previous paper \cite{KK26}, and we omit here. 
\end{proof}

\begin{lemma}\label{BU_lemma 0000}
	Let $b\in \Fqmul$ and assume that $\gcd(r,q-1)\in\{1,2\}$. 
	If $D_{00}(c)\le 1$ for all $c\in \Fqmul$, then $B_{0000}(b)=0$.
\end{lemma}
\begin{proof}
	If $(x,y)\in S_{00}\times S_{00}$, then \eqref{BU_system} leads to 
	\begin{equation*}
		\begin{cases}
			x^r - y^r = \frac{b}{2},\\
			(x+1)^r - (y+1)^r = \frac{b}{2}. 
		\end{cases}
	\end{equation*}
	Then we have $x^r -y^r = \frac{b}{2} = (x+1)^r - (y+1)^r$, which implies $(x+1)^r - x^r = (y+1)^r - y^r=c$ for some $c\in \Fq$. By Lemma \ref{DU_lemma 00}, $D_{00}(0)=0$ and hence $c\ne 0$. Since $D_{00}(c)\le 1$ for all $c\in \Fqmul$, we have $x=y$, and hence $b=0$, a contradiction.
\end{proof}

By applying Lemma~\ref{BU_lemma 0000} to the results of Theorem~\ref{DU_theorem} and Lemma~\ref{DU_lemma r=3 S00}, 
we obtain that $B_{0000}(b)=0$ for every nonzero $b\in\Fq$ 
for all $r$ listed in Table~\ref{table_r}. 

\begin{lemma}\label{BU_lemma1}
	Let $q\equiv 3\pmod{4}$ and $b\in \Fqmul$. Then, $B_{0100}(b) = B_{1000}(b) =0$ or  $B_{0001}(b) = B_{0010}(b) = 0$.
\end{lemma}
\begin{proof}
	Assume that $B_{0100}(b)>0$. Then, \eqref{BU_system} reduces to
	\begin{equation*}
		\begin{cases}
			x^r-y^r = \frac{b}{2},\\
			-(y+1)^r = \frac{b}{2}.
		\end{cases}
	\end{equation*} 
	From the second equation, we have $\chi\left( \frac{b}{2}\right)=-1$.
	
	Suppose that $B_{0001}(b)>0$. Then, the second equation of \eqref{BU_system} reduces to $(x+1)^r = \frac{b}{2}$, and hence there is $x\in S_{00}$ such that $(x+1)^r = \frac{b}{2}$. Then we have $\chi\left( (x+1)^r \right) = 1 = -1 = \chi\left( \frac{b}{2}\right)$, a contradiction. Thus, $B_{0001}(b)=0$. Similarly, if $B_{0010}(b)>0$, then the first equation of \eqref{BU_system} reduces to $x^r = \frac{b}{2}$. Then, we obtain $\chi\left( x^r \right) = 1 = -1 = \chi\left( \frac{b}{2}\right)$, a contradiction, and hence $B_{0010}(b)=0$. 
	
	Similarly, if $B_{1000}(b) >0$, then we have $\chi\left( \frac{b}{2}\right)= \chi(-y^r)=-1$, and hence $B_{0001}(b) = B_{0010}(b) = 0$.
	
	By symmetry, if $B_{0001}(b)>0$ or $B_{0010}(b) >0$, then $B_{0100}(b) = B_{1000}(b) =0$.
\end{proof}

\begin{lemma}\label{BU_lemma2}
	Let $q\equiv 3\pmod{4}$ and $b\in \Fqmul$. If $\gcd(r,q-1)\in \{1,2\}$, then $B_{0001}(b) \le 1$, $B_{0010}(b) \le 1$, $B_{0100}(b)\le 1$, and $B_{1000}(b) \le 1$.
\end{lemma}
\begin{proof}
We only give the proof of $B_{0001}(b)\le 1$, because proofs of the other cases are similar. If $(x,y)\in S_{00}\times S_{01}$, then equations in \eqref{BU_system} are reduced to 
\begin{equation}\label{BU_system S0001}
	\begin{cases}
		x^{r}-y^{r} = \frac{b}{2},\\
		(x+1)^{r} = \frac{b}{2}.
	\end{cases}
\end{equation}
The second equation of \eqref{BU_system S0001} has $\gcd(r,q-1)$ solutions. 
If $\gcd(r,q-1)=1$, then there exists a unique $x\in\Fq$ satisfying it. 
If $\gcd(r,q-1)=2$, then $r$ is even and hence $x^r=(-x)^r$ for all $x\in\Fq$. 
Thus, if $x=\theta$ is a solution of $(x+1)^{r}=\frac{b}{2}$, the other solution is $x=-\theta-2$. 
If both $\theta$ and $-\theta-2$ lie in $S_{00}$, then 
$1=\chi(\theta+1)=\chi(-\theta-1)=-\chi(\theta+1)$, a contradiction. 
Therefore, the second equation of \eqref{BU_system S0001} admits at most one solution in $S_{00}$ in either case.
For such $x$, the first equation of \eqref{BU_system S0001} has at most one solution $y$ with $\chi(y)=1$ in the same manner. 
Hence, $B_{0001}(b)\le1$, which completes the proof.
\end{proof}

Now, we are ready to prove the main theorem of this subsection.

\begin{proof}[Proof of Theorem \ref{BU_theorem}]
	It remains to consider the case where \eqref{BU_system} has a solution with $x=0$ or $y=0$, 
	when $D_{00}(b)\le 1$ for all $b\in \Fqmul$. We only consider the case where \eqref{BU_system} has a solution with $x=0$, 
	since the case $y=0$ can be proved similarly by symmetry. In this case, \eqref{BU_system} leads to 
	\begin{equation}\label{BU_system x=0}
		-(1+\chi(y))y^r = 2-(1+\chi(y+1))(y+1)^r = b.
	\end{equation}
	If \eqref{BU_system x=0} has a solution $y\in \Fq$ with $\chi(y)=-1$ or $y=0$, then we have $b=0$, a contradiction. If \eqref{BU_system x=0} has a solution $y\in S_{00}$, then we have $-2y^r = 2 -2(y+1)^r= b$, so $y\in S_{00}$ is a solution of $(y+1)^r - y^r=1$. Since we already see that $D_{00}(1)\le 1$, there exists at most one $y\in S_{00}$ satisfying $-2y^r = 2 -2(y+1)^r= b$. Let $y=\theta\in S_{00}$ be the solution if it exists. Note that $\chi\left( \frac{b}{2}\right) = \chi(-\theta^r) =-1$. By Lemmas \ref{BU_lemma S11}, \ref{BU_lemma nosol} and \ref{BU_lemma 0000}, 
	\begin{equation}\label{delta1b_inequal x=0}
		\beta_{F_r}(1,b)\le 1 + B_{0001}(b) + B_{0010}(b) + B_{0100}(b) + B_{1000}(b).
	\end{equation}
	If there is another solution of \eqref{BU_system} with $x\in S_{00}$, then we have $x^r = \frac{b}{2}$ or $(x+1)^r=\frac{b}{2}$. If $\gcd(r,q-1)=1$, then we have $x=-\theta$ or $x=-\theta-1$, respectively. But, we observe that $-\theta, -\theta-1\not \in S_{00}$, because $\chi(-\theta)=-1$. If $\gcd(r,q-1)=2$, there is no such $x\in \Fq$ because $r$ is even and $\chi\left( \frac{b}{2}\right)=-1$. Hence, $B_{0001}(b) = B_{0010}(b)=0$. 
	If there is another solution of \eqref{BU_system} in $S_{10}\times S_{00}$, then $y^r=-\frac{b}{2}$. If $\gcd(r,q-1)=1$, then $y=\theta$. Then, $(x+1)^r = (\theta+1)^r +\frac{b}{2} = 1$, and hence $x=0$, a contradiction. If $\gcd(r,q-1)=2$, then two solutions of $y^r=-\frac{b}{2}$ is $y=\pm \theta$, and hence $y=\theta$ since $\chi(-\theta)=-1$. Then, we have $(x+1)^r=1$ which implies that $(x+1)^2=1$. Since $x\ne0$, we obtain $x=-2$. Then, $\chi(x+1)=\chi(-1)=-1$, which is a contradiction to $x\in S_{10}$, and hence $B_{1000}(b)=0$. By Lemma \ref{BU_lemma2} with \eqref{delta1b_inequal x=0}, we obtain $\beta_{F_r}(1,b)\le 2$.
	
	If \eqref{BU_system x=0} has a solution $y\in S_{01}$, then we have $-2y^r = 2 = b$, which implies $y^r=-1$ and hence $r$ is odd. Then we obtain $\gcd(r,q-1)=1$ and $y=-1$, which contradicts that $y\in S_{01}$.
	
	If there is no solution of \eqref{BU_system} with $x=0$ or $y=0$, then by Lemmas \ref{BU_lemma S11}, \ref{BU_lemma nosol} and \ref{BU_lemma 0000}, 
	\begin{equation*}
		\beta_{F_r}(1,b) = B_{0001}(b) + B_{0010}(b) + B_{0100}(b) + B_{1000}(b).
	\end{equation*} 
	Combining Lemma \ref{BU_lemma1} and Lemma \ref{BU_lemma2} with the above equation, we have $\beta_{F_r}(1,b)\le 2$, which completes the proof.
\end{proof}

\subsection{Boomerang Spectrum of $F_2$ when $p=3$}\label{subsec_BS_F2}

In the above subsection, we showed that the boomerang uniformity of $F_r$ is at most~2 for several values of $r$, 
especially when $r = \frac{3^k + 1}{2}$, that is, generalized Coulter-Matthews exponent. 
Next, we consider the special case $k = 1$, that is $r = 2$, and prove that the boomerang uniformity of $F_r$ is actually~$1$ in this case. Moreover, we also provide the full boomerang spectrum of $F_2$.

\begin{remark}\label{MW25_remark}
	The boomerang uniformity of $F_2$ was studied in~\cite{MW25}. 
	Theorem 8 of \cite{MW25} proved that the boomerang uniformity of $F_2$ equals $2$ for $q\ge 9613^2$, and conjectured that this holds for $q>307$. We remark that the conclusion $\beta_{F_2} = 2$ in \cite{MW25} does not hold when $p=3$ and $n$ is odd. Indeed, $\beta_{F_2}=1$ consistently holds in this case, as established in our Theorem \ref{BS_theorem r=2}. 
	The issue is that, in characteristic $3$, the identity $1=-2$ makes certain cases considered in \cite{MW25} incompatible, although they are treated as independent (e.g., conditions arising in Corollary~4).
	This phenomenon is specific to characteristic $3$, where $1=-2$, and does not occur for odd primes $p > 3$.
\end{remark}

By Lemma 18 of \cite{MW25}, \eqref{BU_system} with $r=2$ has no solution with $x=0$ or $y=0$. Hence, as observed in the above subsection, 
$\beta_{F_2}(1,b) = B_{0001}(b) + B_{0010}(b) + B_{0100}(b) + B_{1000}(b)$ holds, 
and that $B_{0001}(b), B_{0010}(b), B_{0100}(b), B_{1000}(b) \le 1$. We next characterize the conditions under which each of these quantities equals~1  
and show that no two of these conditions can hold simultaneously.

\begin{lemma}\label{BS_lemma}
	Let $p=3$, $n$ be odd, $r=2$, and $b\in \Fqmul$. Then, 
	\begin{itemize}
		\item $B_{0001}(b)=1$ if and only if 
		\begin{equation}\label{BU_eqn r=2 0001}
			\chi(b)=\chi\left(b^{\frac{q+1}{4}} + 1 \right) =\chi\left(b^{\frac{q+1}{4}} - 1 \right) = \chi\left( \left( -b^{\frac{q+1}{4}} + 1\right) ^{\frac{q+1}{4}}+1\right) =-1.
		\end{equation}
		\item $B_{0010}(b)=1$ if and only if 
		\begin{equation}\label{BU_eqn r=2 0010}
			\chi(b)
			= \chi\!\left(b^{\frac{q+1}{4}} - 1\right)
			= \chi\!\left(\left(b^{\frac{q+1}{4}} + 1\right)^{\frac{q+1}{4}} - 1\right)
			= -1,\ \chi\!\left(b^{\frac{q+1}{4}} + 1\right)=1.
		\end{equation}
		\item $B_{0100}(b)=1$ if and only if 
		\begin{equation}\label{BU_eqn r=2 0100}
			\chi(b)
			= \chi\!\left(b^{\frac{q+1}{4}} - 1\right)
			=\chi\!\left(b^{\frac{q+1}{4}} + 1\right) = 1, \ \chi\!\left(\left(b^{\frac{q+1}{4}} + 1\right)^{\frac{q+1}{4}} + 1\right)
			= -1.
		\end{equation}
		\item $B_{1000}(b)=1$ if and only if 
		\begin{equation}\label{BU_eqn r=2 1000}
			\chi(b)
			=\chi\!\left(b^{\frac{q+1}{4}} + 1\right) = 1, \
			\chi\!\left(b^{\frac{q+1}{4}} - 1\right) 
			=\chi\!\left(\left(-b^{\frac{q+1}{4}} + 1\right)^{\frac{q+1}{4}} - 1\right)
			= -1.
		\end{equation}
	\end{itemize} 
\end{lemma}

\begin{proof}
	If there is a solution $(x,y)\in S_{00}\times S_{01}$ in \eqref{BU_system} with $r=2$, then \eqref{BU_system} reduces to 
	\begin{equation}\label{BU_system S0001 r=2}
		\begin{cases}
			x^2-y^2 = -b,\\
			(x+1)^2 = -b.
		\end{cases}
	\end{equation} 
	Raising both sides of the second equation in \eqref{BU_system S0001 r=2} to the $\frac{q+1}{4}$-th power yields $x + 1 = -b^{\frac{q+1}{4}}$ and $\chi(b) = -1$, 
	since $\frac{q+1}{4}$ is odd when $q = 3^n$ with $n$ odd and $\chi(x + 1) = 1$.
	Substituting $x = -b^{\frac{q+1}{4}} - 1$ into the first equation in \eqref{BU_system S0001 r=2}, we obtain
	$$
	y^2 = x^2 + b = \left(-b^{\frac{q+1}{4}} - 1\right)^2 + b = -b^{\frac{q+1}{4}} + 1,
	$$
	and hence $y = \left(-b^{\frac{q+1}{4}} + 1\right)^{\frac{q+1}{4}}$. Consequently, existence of solution $(x,y)\in S_{00}\times S_{01}$ satisfying \eqref{BU_system S0001 r=2} implies that $x$ and $y$ are uniquely determined. 
	Furthermore, it is straightforward to check
	$$
	\chi(b)
	= \chi\!\left(b^{\frac{q+1}{4}} + 1\right)
	= \chi\!\left(b^{\frac{q+1}{4}} - 1\right)
	= \chi\!\left(\left(-b^{\frac{q+1}{4}} + 1\right)^{\frac{q+1}{4}} + 1\right)
	= -1.
	$$
	By symmetry, proofs for the other cases are similar and we omit them. 
	
	Conversely, if conditions in \eqref{BU_eqn r=2 0001} are satisfied, then we can easily see that 
	$$
	(x,y) =\left( -b^{\frac{q+1}{4}}-1,\  \left(-b^{\frac{q+1}{4}} + 1\right)^{\frac{q+1}{4}}\right)
	$$ 
	is a solution of \eqref{BU_system S0001 r=2}, by the above process.
\end{proof}

Applying Lemma \ref{Fruproperty2_lemma}, the boomerang spectrum of $F$ is defined to be the multiset $BS_F=\{\nu_i : 0 \le i \le \beta_F\}$, where
\begin{equation*}
	\nu_i = \#\{ b\in \Fqmul : \beta_F(1,b)=i\}.
\end{equation*}
The following identity for the boomerang spectrum is well-known :
\begin{equation}\label{BS_identity}
	\sum_{i=0}^{\beta_F}\nu_i = q-1.
\end{equation}

\begin{theorem}\label{BS_theorem r=2}
	Let $p=3$ and $n\ge 3$ odd. Then $\beta_{F_2}= 1$. Furthermore, the boomerang spectrum of $F_2$ is given by
	\begin{equation*}
		BS_{F_{2}} = \left\{ 
		\nu_0 =  \frac{3q-5+2\Lambda}{4},\quad
		\nu_1 =  \frac{q+1-2\Lambda}{4}
		\right\},
	\end{equation*}
	where $\Lambda = \Sum_{x\in \Fq}\chi(x+1)\chi(x^2+1)$.
\end{theorem}
\begin{proof}
	Since any two of \eqref{BU_eqn r=2 0001}, \eqref{BU_eqn r=2 0010}, \eqref{BU_eqn r=2 0100}, and \eqref{BU_eqn r=2 1000} in Lemma \ref{BS_lemma} cannot occur simultaneously, we obtain $\beta_{F_2}\le 1$.
	
	If $\chi\left(b^{\frac{q+1}{4}} - 1 \right) = -1$, then
	\begin{align*}
		&\chi\left( \left( -b^{\frac{q+1}{4}} + 1\right) ^{\frac{q+1}{4}}+1\right) \chi\left( \left( -b^{\frac{q+1}{4}} + 1\right) ^{\frac{q+1}{4}}-1\right)
		= \chi\left( \left( -b^{\frac{q+1}{4}} + 1\right) ^{\frac{q+1}{2}}-1\right)
		\\
		&=\chi\left( \chi\left( -b^{\frac{q+1}{4}} + 1\right)\cdot \left( -b^{\frac{q+1}{4}} + 1\right) -1\right) = \chi\left(- b^{\frac{q+1}{4}}  \right) = -\chi(b). 
	\end{align*}
	Hence, $B_{0001}(b)=1$ if and only if 
	\begin{align*}
		&\chi(b)=-1,\ 
		\chi\!\left(b^{\frac{q+1}{4}} + 1\right) = -1, \
		\chi\!\left(b^{\frac{q+1}{4}} - 1\right)=-1,\\  
		&-1=\chi\!\left(\left(-b^{\frac{q+1}{4}} + 1\right)^{\frac{q+1}{4}} + 1\right) = -\chi\left( b \right)\chi\!\left(\left(-b^{\frac{q+1}{4}} + 1\right)^{\frac{q+1}{4}} - 1\right) = \chi\!\left(\left(-b^{\frac{q+1}{4}} + 1\right)^{\frac{q+1}{4}} - 1\right).
	\end{align*}
	Thus, $B_{0001}(b)+B_{1000}(b)=1$ if and only if 
	\begin{equation}\label{BU_eqn r=2 0001 1000}
		\chi(b)\chi\left(b^{\frac{q+1}{4}} + 1 \right) =1,\ 
		\chi\left(b^{\frac{q+1}{4}} - 1 \right) 
		= \chi\left( \left( -b^{\frac{q+1}{4}} + 1\right) ^{\frac{q+1}{4}} - 1\right) =-1.
	\end{equation}
	Similarly, $B_{0100}(b)+B_{0010}(b)=1$ if and only if 
	\begin{equation}\label{BU_eqn r=2 0010 0100}
		\chi(b)\chi\left(b^{\frac{q+1}{4}} - 1 \right) = \chi\left(b^{\frac{q+1}{4}} + 1 \right) = 1,\ \chi\left( \left( b^{\frac{q+1}{4}} + 1\right) ^{\frac{q+1}{4}}-1\right) =-1.
	\end{equation}
	
	\eqref{BU_eqn r=2 0010 0100} coincides with \eqref{BU_eqn r=2 0001 1000} after substituting $b$ by $-b$. Hence, $\nu_1$ equals twice the number of $b$ satisfying \eqref{BU_eqn r=2 0010 0100}, and is therefore given by
	\begin{equation}\label{nu1_eqn}
		\nu_1 =\Frac{1}{4}\Sum_{b\in \Fq\setminus \F_3}\left( 1 + \chi(b)\chi\!\left(b^{\frac{q+1}{4}} - 1\right)\right) 
		\left( 1 + \chi\!\left(b^{\frac{q+1}{4}} + 1\right)\right) 
		\left( 1 - \chi\left( \left( b^{\frac{q+1}{4}} + 1\right) ^{\frac{q+1}{4}}-1\right)\right)  = \Frac{1}{4}(A-2),
	\end{equation}
	where
	\begin{equation*}
		A=\Sum_{b\in \Fq}\left( 1 + \chi(b)\chi\!\left(b^{\frac{q+1}{4}} - 1\right)\right) 
		\left( 1 + \chi\!\left(b^{\frac{q+1}{4}} + 1\right)\right) 
		\left( 1 - \chi\left( \left( b^{\frac{q+1}{4}} + 1\right) ^{\frac{q+1}{4}}-1\right)\right).
	\end{equation*}
	Let $y=b^{\frac{q+1}{4}}$. Then, since $\frac{q+1}{4}$ is odd, applying Lemma \ref{chi_lemma quad} and Lemma \ref{chi_lemma cubic},
	\begin{align*}
		A& = q +\Sum_{y\in \Fq}\chi(y(y-1)) + \Sum_{y\in \Fq} \chi(y+1) 
		- \Sum_{y\in \Fq} \chi\left( (y+1) ^{\frac{q+1}{4}}-1\right)
		+ \Sum_{y\in \Fq} \chi(y(y^2-1)) \\
		&-\Sum_{y\in \Fq}\left( \chi(y(y-1)) +\chi(y+1) +\chi(y(y^2-1))\right) \chi\left( (y+1) ^{\frac{q+1}{4}}-1\right)\\
		&=q - 1 -\Sum_{y\in \Fq}\left( \chi(y(y-1)) +\chi(y+1) +\chi(y(y^2-1))\right) \chi\left( (y+1) ^{\frac{q+1}{4}}-1\right)\\
		&=q - 2 -\Sum_{y\in \Fq\setminus\{-1\}}\left( \chi(y(y-1)) +\chi(y+1) +\chi(y(y^2-1))\right) \chi\left( (y+1) ^{\frac{q+1}{4}}-1\right).
	\end{align*}
	If $t=(y+1)^{\frac{q+1}{4}}$, then 
	\begin{equation*}
		y = 
		\begin{cases}
			t^2-1 &\text{ if }\chi(t)=1,\\
			-t^2-1 &\text{ if }\chi(t)=-1.
		\end{cases}
	\end{equation*}
	Hence,
	\begin{align*}
		A &= q -2- \Sum_{\chi(t)=1}\left( \chi((t^2-1)(t^2-2)) +\chi(t^2) +\chi(t^2(t^2-1)(t^2-2))\right) \chi\left( t-1\right)\\
		&- \Sum_{\chi(t)=-1}\left( \chi((-t^2-1)(-t^2-2)) +\chi(-t^2) +\chi(-t^2(-t^2-1)(-t^2-2))\right) \chi\left( t-1\right)\\
		&= q -2 -\Sum_{\chi(t)=1}\left( 2\chi((t^2-1)(t^2-2)) +1 \right) \chi\left( t-1\right) + \Sum_{\chi(t)=-1}\chi(t-1).
	\end{align*}
	Since 
	\begin{equation*}
		\Sum_{\chi(t)=1}\chi(t-1)-\Sum_{\chi(t)=-1}\chi(t-1)= \frac{1}{2}\left( \Sum_{t\in \Fqmul}  (1+\chi(t))\chi(t-1)-\Sum_{t\in \Fqmul}(1-\chi(t))\chi(t-1) \right) = \Sum_{t\in \Fq}  \chi(t(t-1)) =-1,
	\end{equation*}
	we obtain
	\begin{align*}
		A&= q  - 1 - 2\Sum_{\chi(t)=1} \chi((t^2-1)(t^2-2)(t-1)) 
		= q - 1 - 2\Sum_{\substack{\chi(t)=1\\ t\ne 1}} \chi((t+1)(t^2-2)(t-1)^2)\\
		&= q - 1 - 2\Sum_{\substack{\chi(t)=1\\ t\ne 1}} \chi((t+1)(t^2-2))
		= q + 1 - 2\Sum_{\chi(t)=1} \chi((t+1)(t^2-2)) \\
		&= q + 1 - \Sum_{t\in \Fqmul} (1+\chi(t))\chi((t+1)(t^2+1)) \\
		&= q + 1 - \Sum_{t\in \Fqmul} \chi((t+1)(t^2+1))- \Sum_{t\in \Fq} \chi(t(t+1)(t^2+1)).
	\end{align*}
		Let $z=\frac{t(t+1)}{t^2+1}$. Then, $z=1$ if and only if $t=1$. For $t\ne1$, then $(z-1)t^2-t+z=0$. The discriminant of this quadratic equation on $t$ can be computed as $(-1)^2-z(z-1) = 1+z-z^2$. Hence, applying $z=-s-1$, we obtain
	\begin{align*}
		\Sum_{t\in \Fq} \chi(t(t+1)(t^2+1)) &= 1 + \Sum_{t \in \Fq\setminus \{1\}}\chi(t(t+1)(t^2+1)) = 1 + \Sum_{t \in \Fq\setminus \{1\}}\chi\left( \Frac{t(t+1)}{t^2+1}\right) \\
		&= 1+\Sum_{z \in \Fq\setminus \{1\}} (1+\chi(1+z-z^2))\chi(z) = -1+\Sum_{z\in \Fq}\chi(z)\chi(1+z-z^2)\\
		&= -1+\Sum_{s\in \Fq}\chi(-s-1)\chi(-s^2-1) = -1+\Sum_{s\in \Fq}\chi(s+1)\chi(s^2+1)
	\end{align*}	
	Thus, we get
	\begin{align*}
		A&= q + 2 - \Sum_{t\in \Fq} \chi((t+1)(t^2+1)) - \left( -1 +\Sum_{t\in \Fq} \chi((t+1)(t^2+1))\right) \\
		&= q + 3 - 2\Sum_{t\in \Fq} \chi((t+1)(t^2+1)).
	\end{align*}
	Hence, \eqref{nu1_eqn} reduces to 
	\begin{equation*}
		\nu_1 =\Frac{1}{4}\left(  A-2  \right)  
		= \Frac{1}{4} \left( q +1  - 2\Sum_{x\in \Fq}\chi((x+1)(x^2+1))\right) = \Frac{1}{4}(q+1-2\Lambda)
	\end{equation*}
	as desired. Applying \eqref{BS_identity}, we have
	\begin{equation*}
		\nu_0 =q-1-\nu_1 =  \Frac{1}{4} \left( 3q -5  + 2\Lambda\right).
	\end{equation*}
	
	By Lemma \ref{chi_lemma inequal}, $|\Lambda| \le 2\sqrt{q}$. Hence, we have $\nu_1 =  \frac{1}{4}(q+1-2\Lambda) \ge \frac{1}{4}(q+1-4\sqrt{q})$. We confirm that $\frac{1}{4}(q+1-4\sqrt{q})>1$ when $q>22$, and therefore $\beta_{F_2}=1$ when $n\ge 3$ is odd. 
\end{proof}

The boomerang spectrum of $F_2$ for $3 \le n \le 9$, verified using SageMath, is shown in Table~\ref{table_BS_F2}.

\begin{table}[h!]
	\centering
	\setlength{\tabcolsep}{10pt}
	\renewcommand{\arraystretch}{1.25}
	\begin{tabular}{ccc}
		\hline
		$n$ & $\Lambda$ & $BS_{F_2}$ \\
		\hline
		3        & $-10$   & $\{ \nu_0 =14,\ \nu_1 =12\}$     \\
		5       & $2$     & $\{ \nu_0 =182,\ \nu_1 =60\}$      \\
		7      & $86$    & $\{ \nu_0 =1682,\ \nu_1 =504\}$       \\
		9     & $-190$  & $\{ \nu_0 =14666,\ \nu_1 =5016\}$     \\
		\hline
	\end{tabular}
	\caption{Boomerang spectrum $BS_{F_2}$ when $p=3$, $n$ is odd, $3\leq n \leq 9$.}\label{table_BS_F2}
\end{table}

By Lemma~\ref{CM linear equiv} and Theorem~\ref{BS_theorem r=2}, the following result follows immediately.

\begin{corollary}\label{BS_corollary r=2}
	Let $p=3$ and $n\ge 3$ be odd. 
	Then, $\beta_{F_{\frac{3^{n-1}+1}{2}}}=1$. 
	Furthermore, $F_{\frac{3^{n-1}+1}{2}}$ has the same boomerang spectrum as $F_2$, 
	which is given in Theorem~\ref{BS_theorem r=2}.
\end{corollary}

\section{Conclusion}\label{sec_con}

In this paper, we show that if $\gcd(q-1,r)\mid 2$ and $D_{00}(b)\le 1$ for all $b\in \Fqmul$, then $F_r$ is locally-APN with boomerang uniformity at most~$2$. 
Moreover, it holds that $D_{00}(b)\le 1$ for all $b\in\Fqmul$ for the exponents~$r$ listed in Table~\ref{table_r}.
We further investigate the differential spectra of $F_3$ and $F_{\frac{2q-1}{3}}$. 
In addition, by analyzing the boomerang spectrum of $F_2$ when $p=3$, we prove that its boomerang uniformity is equal to~$1$, 
which shows that the claimed conclusion in~\cite{MW25} does not hold in characteristic $3$, where it was claimed that the boomerang uniformity of $F_2$ equals~$2$ for sufficiently large~$q$. 

It should be noted that Table~\ref{table_r} does not exhaust all APN power functions 
with $\gcd(r,q-1)=2$. 
Preliminary experiments suggest that $D_{00}(b)\le1$ may also hold for some additional exponents, 
which we plan to investigate in future work.
We also plan to further explore how this idea can be extended in future work.

\bigskip
\noindent \textbf{Acknowledgements}:
The authors would like to thank the anonymous reviewers for their kind comments and valuable suggestions. This work was supported by the National Research Foundation of Korea (NRF) grant funded by the Korea government (MSIT) (No. RS-2021-NR061794). Soonhak Kwon was supported by Basic Science Research Program through the National Research Foundation of Korea (NRF) funded by the Ministry of Education (No. RS-2019-NR040081).

\end{document}